\documentclass[letterpaper, 10 pt, conference]{ieeeconf}

\IEEEoverridecommandlockouts                              

\usepackage{tikz}
\usepackage{dsfont}
\usepackage{amsfonts}
\usepackage{amsmath}
\usepackage{amssymb}
\usepackage{color}
\usepackage{graphicx}
\usepackage[mathscr]{eucal}
\newtheorem{lemma}{Lemma}
\newtheorem{proposition}{Proposition}
\newtheorem{corollary}{Corollary}
\newtheorem{fact}{Fact}

\newtheorem{remark}{Remark}

\newtheorem{assumption}{Assumption}

\def\begcen{\begin{center}}
\def\endcen{\end{center}}

\newcommand{\bfz}{\mbox{$z$}}

\def\bfz{{\bf z}}

\newcommand{\col}{ \mbox{col} }

\def\calp{{\cal P}}

\def\hal{{1 \over 2}}

\def\L2{{\cal L}_2}
\def\L2e{{\cal L}_{2e}}

\def\rea{\mathds{R}}

\def\begmat#1{\begin{bmatrix}#1\end{bmatrix}}

\def\begalis#1{\begin{align*}{#1}\end{align*}}

\def\begequarr{\begin{eqnarray}}
\def\endequarr{\end{eqnarray}}
\def\begequarrs{\begin{eqnarray*}}
\def\endequarrs{\end{eqnarray*}}
\def\begarr{\begin{array}}
\def\endarr{\end{array}}
\def\begequ{\begin{equation}}
\def\endequ{\end{equation}}
\def\lab{\label}
\def\begdes{\begin{description}}
\def\enddes{\end{description}}
\def\begenu{\begin{enumerate}}
\def\begite{\begin{itemize}}
\def\endite{\end{itemize}}
\def\endenu{\end{enumerate}}

\def\lef[{\left[\begin{array}}
\def\rig]{\end{array}\right]}

\def\begcen{\begin{center}}
\def\endcen{\end{center}}
\def\begrem{\begin{remark}\rm}
\def\endrem{\end{remark}}
\def\begassum{\begin{assumption}}
\def\endassum{\end{assumption}}
\def\begassums{\begin{assumption*}}
\def\endassums{\end{assumption*}}
\def\begassu{\begin{ass}}
\def\endassu{\end{ass}}
\def\beglem{\begin{lemma}}
\def\endlem{\end{lemma}}
\def\begcor{\begin{corollary}}
\def\endcor{\end{corollary}}
\def\begfac{\begin{fact}}
\def\endfac{\end{fact}}

\def\AUT{{\it Automatica}}

\def\ba{\begin{array}}
\def\ea{\end{array}}
\def\begsubequ{\begin{subequations}}
\def\endsubequ{\end{subequations}}
\def\bfpsi{{\boldsymbol{\Psi}}}
\def\bfeps{{\boldsymbol \epsilon}}

\def\bfz{{\bf Z}}

\def\calp{{\mathfrak p}}

\title{An Observer-Based Composite Identifier for Online Estimation of the Th\'evenin Equivalent Parameters of a Power System}
\author{Daniele Zonetti, Romeo Ortega, Rafael Cisneros, Alexey Bobtsov, Fernando Mancilla-David, Oriol Gomis-Bellmunt\thanks{This work was supported by the Ministry of
Science and Higher Education of Russian Federation, passport of goszadanie no. 2019-0898 and by FEDER/Ministerio de Ciencia, Innovaci\'{o}n y Universidades-Agencia Estatal de Investigaci\'{o}n, Project RTI2018-095429-B-I00. The work of O. Gomis-Bellmunt is supported by the ICREA Academia program.}
\thanks{D. Zonetti and O. Gomis-Bellmunt are with the Centre  d'Innovaci\'{o} Tecnol\`{o}gica en Convertidors Est\`{a}tics i Accionaments, Departament d'Enginyeria El\`{e}ctrica, Universitat Polit\`{e}cnica de Catalunya, Barcelona 08028, Spain. (e-mail: daniele.zonetti(oriol.gomis)@upc.edu). }
\thanks{R. Ortega and R. Cisneros are with the Departamento Acad\'{e}mico de Sistemas
Digitales, ITAM, Rio Hondo 1, Col. Progreso Tizapan, 01080 Ciudad de
M\'{e}xico, Mexico (email: romeo.ortega(rcisneros)@itam.mx).}
\thanks{A. Bobtsov is with the Department of Control Systems and Robotics, ITMO University, Kronverkskiy av. 49, Saint-Petersburg, 197101, Russia (e-mail: bobtsov@mail.ru).}
\thanks{F. Mancilla--David is with the Department of Electrical Engineering, University of Colorado Denver, Denver, Colorado 80204, USA, (e-mail: Fernando.Mancilla-David@ucdenver.edu)}}
\begin{document}

\maketitle
\begin{abstract}
    We consider a Thévenin equivalent circuit capturing the dynamics of a power grid as seen from the point of common coupling with a power electronic converter, and provide a solution to the problem of \textit{online} identification of the corresponding circuit parameters. For this purpose, we first derive a linear regression model in the conventional $\tt abc$ coordinates and next design a bounded observer-based composite identifier that requires local measurements and knowledge of the grid frequency only. An extension that guarantees exponential convergence of the estimates, under the additional assumption of knowledge of the grid X/R ratio, is further provided. The performance of the proposed identifier, which subsumes a conventional gradient descent algorithm, is illustrated via detailed computer simulations.
\end{abstract}
\section{Introduction}
\lab{sec1}
With the widespread penetration of power electronic converters in the existing power systems, the problem of establishing an accurate model that is representative of the grid dynamics for the purpose of analysis and of control design is timely and increasingly relevant \cite{GOMIS}. The conventional approach is to use a Thévenin equivalent (TE) circuit to represent the overall grid as seen from the individual bus where the power converter is interfaced. The theoretical basis standing behind this concept is provided by the TE theorem, stating in its original version that any linear circuit constituted by current and/or voltage sources and resistances can be described via an equivalent circuit characterized by a voltage source combined in series with a resistance---a formulation that has been later extended to single-phase and three-phase AC circuits, with the notion of resistance being replaced by the notion of impedance~\cite{DORFbook}. This representation is suited for the analysis of a variety of key issues for a correct and safe operation of the power converter, ranging from  fault studies~\cite{Shen}, voltage stability analysis~\cite{Sun}, loadability limit computation~\cite{Peng} to tuning of the related controllers~\cite{Raghami}. In these studies it is emphasized that, to make the problem tractable in practice, the identification of the TE parameters shall be accomplished using only local measurements of voltage and currents.\\
The problem of online estimation of the TE parameters has been widely studied in the literature mainly using least-squares or extended Kalman filter techniques. As is well-known, the excitation requirements of these estimators, namely persistent excitation \cite{SASBODbook}, is quite stringent and hard to satisfy (without probing signals) in the current application. Another disadvantage of these methods is that---if the excitation conditions are satisfied---the adaptation gain converges to zero losing the alertness of the estimator. In \cite{ARAetal} a variation of the least-squares method proposed 32 years ago in \cite{ORT}, that converges in finite time, was used for this problem. Although the excitation conditions are weak the algorithm---that in its initial stage is akin to an off-line estimator---involves a numerically sensitive matrix inversion and, as it converges to a standard least-squares, loses its alertness.\\

The contributions of the paper are as follows. In Section \ref{sec2} we discuss the fundamental assumptions behind the formulation of the mathematical model of a power converter interfaced to a TE circuit. Then, in Section \ref{sec3} we derive a LRE that is instrumental for the design of a bounded  observer-based composite identifier, a task that is carried-out in Section \ref{sec4}. By employing a further assumption, we next generate in Section~\ref{sec5} an alternative, reduced LRE, for which a conventional gradient descent algorithm can be applied to ensure exponential convergence.  The usefulness of the theoretical results are illustrated via simulations in Section \ref{sec6}. We conclude the paper in Section~\ref{sec7} with some final remarks and guidelines for future investigation.\smallbreak


\noindent {\bf Notation.}  Given a vector $a \in \rea^n$, we denote the square of the Euclidean norm as $|a|^2:=a^\top a$. The symbol $\mathsf{1}_3$ denotes a three-dimensional vector of ones. Given a differentiable signal $u(t) \in \rea^r$, we define the derivative operator \mbox{${du(t)\over dt}=:\calp [u(t)]$} and denote the action of an LTI filter $F(\calp) \in \rea(\calp)$  as $F(\calp)[u(t)]$. For given $\omega>0$, $\varphi\in\mathbb{R}$, we denote the vector
$$
\mathcal{S}_\varphi (t):={\sqrt{\frac{2}{ 3}}} \begin{bmatrix}
\sin(\omega t+\varphi)\\
\sin(\omega t+\varphi-\frac{2}{3}\pi)\\
\sin(\omega t+\varphi+\frac{2}{3}\pi)
\end{bmatrix}. 
$$
%
\section{Assumptions}
\lab{sec2}
The mathematical model, in $\tt abc$ reference frame, of the TE of a three-phase, symmetrical configured system---see Fig. \ref{fig:scheme} for the corresponding circuit schematic---is given by\footnote{A
three-phase AC electrical system is said to be symmetrically
configured if a symmetrical feeding voltage yields a symmetrical
current and vice versa~\cite{DORFbook}.}
\begin{equation}\label{sysabc}
L \frac{d}{dt}i =   -R i + v - e,
\end{equation}
where: $i(t)\in\mathbb{R}^3$, $e(t) \in\mathbb{R}^3$ denote the three--phase current and equivalent voltage of the grid, respectively; $v(t)\in\mathbb{R}^3$ denotes the three-phase voltage synthesized by the converter at the point of common coupling (PCC); $L\in\mathbb{R}_{>0}$ and $R\in\mathbb{R}_{>0}$  denote respectively the grid equivalent inductance and resistance.\smallbreak

\begin{figure}
\centering
\includegraphics[scale=0.2]{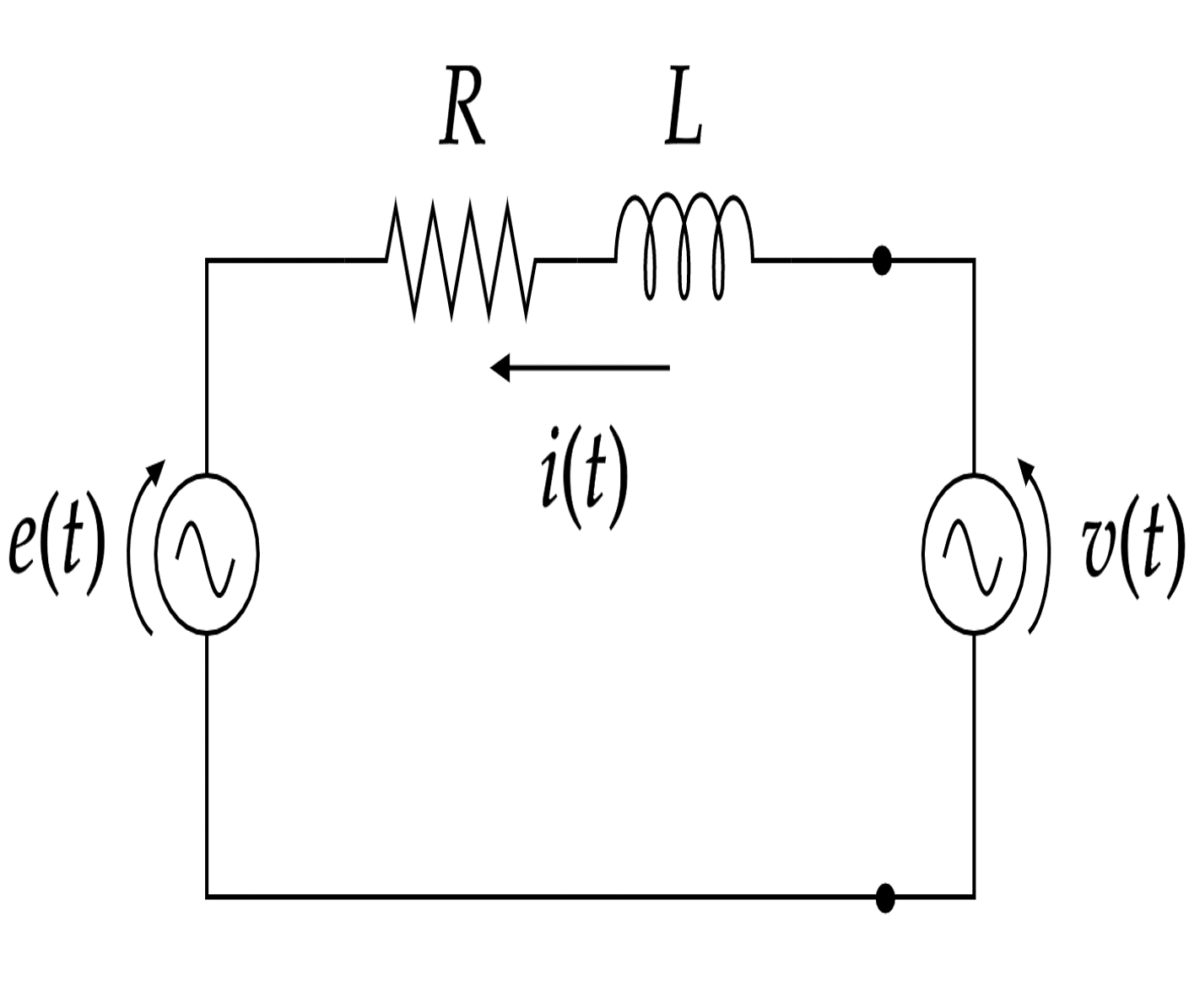}
\caption{Single-line circuit schematic of a power converter connected to a TE grid.}
\label{fig:scheme}
\end{figure}

\begin{assumption}[Grid voltage]
\label{ass1} 
The TE voltage source is described by a three-phase balanced purely sinusoidal signal:
\begin{equation}
\lab{vg}
e = E \mathcal{S}_{0},
\end{equation}
with $E>0$.
\end{assumption}\smallbreak

\begin{assumption}[Parameters]
\lab{ass2}
The frequency $\omega>0$ is \textit{known} and the positive parameters $L$, $R$, $E$ are \textit{unknown} constants.
\end{assumption} \smallbreak

\begin{assumption}[Measurements]
\lab{ass4}
The three-phase signals $i(t)\in \rea^3$ and $v(t)\in \rea^3$ are {\em measurable}.
\end{assumption}\smallbreak

Assumption~\ref{ass1} is justified for AC grids that are characterized by a sufficiently large number of synchronous rotating machines, providing large inertia, and/or whenever a sufficiently tight regulation of the frequency is implemented via grid-forming controllers.
Assumption \ref{ass2} can be instead enforced via the design of a suitably defined phase-locked loop (PLL), which allows to recover the exact value of the grid frequency at a faster time-scale.  Assumption~\ref{ass4} is quite natural and verified in all practical scenarios, as those measurements are utilized for the control of the converter.\\

For an appropriate modeling of the three-phase voltage synthesized by the converter at the PCC a caveat is necessary. In order to guarantee a correct operation of the power system, the state and input variables  $\{i,v\}$ are required to converge to suitable balanced, three-phase AC signal with the common grid frequency $\omega$, but different amplitudes and phase shifts. These operating conditions are enforced by the converter control algorithms, which are usually based on a $\tt dq$ reference frame ensuring that the steady-states of interest, in new coordinates, correspond to constant quantities.  Instrumental for the definition of such reference frame is the knowledge of the grid frequency $\omega$. Indeed, using this information the overall system composed by the TE circuit and the power converter---transformed in $\tt dq$ coordinates---is linear time-invariant (LTI) and adequate control solutions can be established---see for example~\cite{ZONBOBORT}. The design of an exponentially stabilizing controller legitimizes then the following assumption. \smallbreak

\begin{assumption}[PCC voltage]
\label{ass3} 
The voltage at the PCC is described by the three-phase signal
\begin{equation}\lab{v}
v = V \mathcal{S}_\phi+\bfeps_{t} ,
\end{equation}
with $V>0$, $\phi\in(-\frac{\pi}{2},\frac{\pi}{2})$, where $\bfeps_{ t}\in\rea^3$ is a signal exponentially decaying to zero.
\end{assumption}\smallbreak

\begin{remark}
A different steady-state amplitude $V$, phase $\phi$ and exponential term $\bfeps_{t}$ is triggered in \eqref{v} \textit{any time} a change occurs in the parameters of the TE or in the references provided to the power converter controllers.  The rapidity at which the signal $\bfeps_{t}$ vanishes is determined by the tuning of the converter control algorithms.
\end{remark} 

\section{Derivation of a LRE}
\lab{sec3}
Following the standard procedure for parameter estimation \cite{LJUbook,SASBODbook}, we derive a LRE for the system \eqref{sysabc}, as explained in the following lemma.\smallbreak
%
\begin{lemma} 
\lab{lem1}
Consider the system \eqref{sysabc} verifying  {\bf Assumptions A1-A4}. There exist {\em measurable} signals  $\bfz(t)\in\rea^3$ and $\bfpsi_f(t) \in\rea^{3 \times 3}$ such that the following LRE is satisfied
\begin{equation}\lab{lre}
\bfz=\bfpsi_f \theta + \bfeps_t,
\end{equation}
where
\begin{equation}\lab{barthe}
\theta:=\col\left({R \over L}, {1 \over L}, {E \over L} \right)\in \rea^3,
\end{equation}
and $\bfeps_t\in\rea^3$ is a signal exponentially decaying to zero.
\end{lemma}

\begin{proof}
The equation \eqref{sysabc} may be written as
\begin{equation}\lab{doti}
\frac{d}{dt}i =\bfpsi \theta,
\end{equation}
where we defined the three-dimensional square matrix
\begequ
\lab{pfpsi}
\bfpsi:=\begmat{-i &| & v & | & -\mathcal{S}_0 }.
\endequ
The proof is completed applying to \eqref{doti} the LTI, stable filter
$$
	F(\calp)={\lambda \over \calp+\lambda},
$$
with $\lambda>0$ and defining
\begalis{
\bfz(t) & :=\calp F(\calp)[i(t)],\\
\bfpsi_f(t)&:= F(\calp)[\bfpsi(t)].
}
\end{proof}

\begin{remark}\label{rem:not-PE}
Since we assumed that the power system configuration is symmetric, the sum of three phases of both the grid current and the voltage at the PCC is always zero, that is:
\begin{equation}\label{eq:sum-phases}
\mathsf{1}_3^\top i_{g}(t)=\mathsf{1}_3^\top v(t)=\mathsf{1}_3^\top \mathcal{S}_0(t)=0,\quad \forall t\geq 0.    
\end{equation}
A consequence of this fact is that the regression matrix $\bfpsi$ satisfies $\bfpsi\mathsf{1}_3=0$, hence it is singular and cannot satisfy the persistent excitation requirement
\begin{equation}\label{PEcond}
	\int_{t}^{t+T}\bfpsi(s)\bfpsi^\top(s)ds \geq \delta I_3,\;\forall t \geq 0,
\end{equation}
and some $T>0$ and $\delta>0$. Note also that in view of \eqref{eq:sum-phases}, only two of the three $\tt abc$ phases are required for the construction of such matrices.
\end{remark}

\section{An Observer-Based Composite Identifier}
\lab{sec4}
Clearly, a classical gradient descent (or least squares) algorithm can be applied to the LRE \eqref{lre}. However, in view of the lack of excitation indicated in Remark \ref{rem:not-PE}, poor performances are observed in simulations, that cannot be further improved by appropriate tuning of the estimator. 

In view of this situation we propose instead the use of the composite identifier of \cite{PANORTMOY}, which offers additional freedom in the design at the expense of an higher-order dynamics.
\smallbreak

\begin{proposition} 
\lab{pro1}
Consider the system \eqref{sysabc} and the LRE (\ref{lre}). Define the observer-based composite identifier:
\begin{equation}\label{ident}
\begin{aligned}
\dot{\hat i} & = - \alpha (\hat i-i)+ \bfpsi \hat \theta\\
\dot{\hat {\theta}} & = -\gamma_P\bfpsi^\top (\hat i-i)+ \gamma_I   \bfpsi_f^\top( \bfz - \bfpsi_f \hat \theta),
\end{aligned}
\end{equation}
with tuning gains $\alpha>0$, $\gamma_P>0$ and $\gamma_I>0$. Then the system trajectories are globally bounded and 
\begin{equation}\label{eq:convergence-x}
\lim_{t \to \infty}\hat {i}(t)=i(t).
\end{equation}
\end{proposition}\smallbreak

\begin{proof}
Let us define the incremental variables
$$\tilde i:=\hat i - i,\quad \tilde \theta:=\hat\theta - \theta.$$
Using \eqref{doti} and \eqref{lre} we obtain the error dynamics:
\begalis{
\dot{\tilde i} & = - \alpha \tilde i+ \bfpsi \tilde \theta\\
\dot{\tilde {\theta}} & = -\gamma_P \bfpsi^\top \tilde i- \gamma_I  \bfpsi_f^\top \bfpsi_f \tilde \theta+\bfpsi_f^\top\bfeps_t.
}
Consider the Lyapunov function candidate
$$
\mathcal V(\tilde i, \tilde \theta)= \hal |\tilde i|^2 + \frac{1}{2 \gamma_P}|\tilde \theta|^2.
$$
Some simple calculations show that
$$
\dot{\mathcal V}=-\alpha |\tilde i|^2 - \frac{\gamma_I}{\gamma_P}|\bfpsi_f \tilde \theta|^2+\frac{1}{\gamma_P}\tilde\theta^\top  \bfpsi_f^\top\bfeps_t.
$$ 
The proof is completed invoking \cite[Proposition 3.1]{PANORTMOY}.
\end{proof}\smallbreak
 
\begin{remark}\label{rem2}
If $\alpha=\gamma_P=0$, the observer-based composite identifier \eqref{ident} boils down to a classical gradient descent algorithm. For both solutions it would suffice the regression matrix $\bfpsi_f$ to be \textit{persistently exciting} to guarantee convergence of $\hat\theta(t)$ to the actual vector of parameters $\theta$---a condition that is unfortunately not verified, see Remark~\ref{rem:not-PE}. However, as it will be shown in Section~\ref{sec6}, a sufficiently small estimation error can be ensured by appropriate tuning of the parameters. 
\end{remark}

\section{Derivation of a Reduced LRE}\label{sec5}
In this section we derive a second-order LRE that can be used to solve to the problem of online identification of the TE circuit parameters under the following additional assumption. \smallbreak

\begin{assumption}
$R=0$ or, alternatively, the ratio $\rho:= L\omega/R$ is a \textit{known} constant.
\end{assumption}\smallbreak

This assumption can be justified whenever the power grid approximated by the TE circuit is dominantly inductive or it is characterized by transmission lines with a similar $X/R$ ratio---the latter implying that at the occurrence of a trip of one or several lines, the average ratio remains unchanged. A straightforward consequence of adding the {\bf Assumption A5} is that we can reduce the number of parameters to be estimated from three to two. Indeed, from the knowledge of the grid inductance $L$ we can compute the grid resistance from the equation $R=L\omega/\rho.$ \smallbreak

\begin{proposition}\label{proposition2}
 Consider the system \eqref{sysabc} verifying  {\bf Assumptions A1-A5}. There exist {\em measurable} signals  $\bfz_{\tt ab}(t)\in\rea^2$ and $\bfpsi_{f,\tt{ab}}(t) \in\rea^{2 \times 2}$ such that the following LRE is satisfied
\begin{equation}\lab{lre-ab}
\bfz_{\tt ab}=\bfpsi_{f,\tt{ab}} \vartheta + \bfeps_t,
\end{equation}
where
\begin{equation}\lab{barthe-ab}
\vartheta:=\col\left({1 \over L}, {E \over L} \right)\in \rea^2,
\end{equation}
and $\bfeps_t\in\rea^2$ is a signal exponentially decaying to zero.\\
Moreover, if $\phi\neq 0$, the trajectories generated by the gradient descent estimator
\begin{equation}\label{eq:GD}
    \dot{\hat\vartheta}=\gamma \bfpsi_{f,\tt ab}^\top(\bfz_{\tt ab}-\bfpsi_{f,\tt ab}\hat\vartheta),
\end{equation}
with tuning gain $\gamma>0$, verify
$$
\lim_{t\to\infty}\hat\vartheta(t)=\vartheta,\quad (exp).
$$
\end{proposition}\smallbreak

\begin{proof}
Let us consider the phase $\tt a$ and $\tt b$ from equation \eqref{sysabc}, so that we can write\footnote{Phases $a$ and $b$ are selected with no loss of generality, since a similar proof applies for a different choice of the phases.}
\begin{equation}\lab{doti-ab}
\frac{d}{dt}i_{\tt ab}+\frac{\omega}{\rho}i_{\tt ab} =\bfpsi_{\tt ab} \vartheta,
\end{equation}
where we defined the two-dimensional square matrix
\begin{equation}\lab{pfpsi-ab}
\bfpsi_{\tt ab}:=\begmat{v_{\tt ab} & | & -\mathcal{S}_{0,\tt{ab}} }.
\end{equation}
The first part of the proof is then completed, similar to Lemma~\ref{lem1}, by applying to \eqref{doti-ab} the LTI, stable filter
$$
	F(\calp)={\lambda \over \calp+\lambda},
$$
with $\lambda>0$ and defining
\begalis{
\bfz_{\tt ab}(t) & :=\Big(\calp+{\omega \over \rho}\Big)F(\calp)[i_{\tt ab}(t)]\\ 
\bfpsi_{f,\tt{ab}}(t)&:= F(\calp)[\bfpsi_{\tt ab}(t)].
}

To prove global exponential convergence, recall that
$$
\lim_{t\to\infty}\bfpsi_{\tt ab}=[V\mathcal S_{\phi,\tt ab}\;|\;-\mathcal{S}_{0,\tt ab}]
$$
and that, using standard trigonometric identities, we have

\begin{equation}\resizebox{0.49\textwidth}{!}{$
\begin{aligned}
&\det [V\mathcal S_{\phi,\tt ab}\;|\;-\mathcal{S}_{0,\tt ab}]= \\
&=V\left(\sin(\omega t)\sin(\omega t+\phi-\frac{2}{3}\pi)-\sin(\omega t+\phi)\sin(\omega t-\frac{2}{3}\pi)\right)\\
&=\frac{V}{2}\left(\cos(\phi-\frac{2}{3}\pi)-\cos(\phi+\frac{2}{3}\pi)\right)\\
&=\frac{\sqrt{3}}{2}V \sin(\phi),
\end{aligned}$}
\end{equation}
which is always non-zero, being $\phi\neq 0.$ Then we can claim that the matrix $\bfpsi_{\tt ab}$  eventually converges to a nonsingular matrix and that, since this property is preserved upon stable filtering, so does the regression matrix $\bfpsi_{f,\tt ab}$. By leveraging standard theory in parameters identification~\cite{NARENDRAbook} we conclude that $\bfpsi_{f,\tt ab}$ is persistently exciting and therefore global exponential convergence can be ensured via the gradient descent estimator~\eqref{eq:GD}.
\end{proof}\smallbreak

\begin{remark}
The property of persistency of excitation of the regression matrix $\bfpsi_{f,\tt ab}$ is guaranteed under the mild assumption $\phi\neq 0$, which corresponds to having a non-zero active power transferred between the power converter and the grid.
\end{remark}\smallbreak
 
 \begin{remark}
It is clear from the proof of Proposition \ref{pro1} that an observer-based composite identifier analogous to \eqref{ident} can be applied to the LRE \eqref{lre-ab}, at the cost of an higher-order dynamics. However, in view of the persistency of excitation of the regression matrix $\bfpsi_{f,\tt ab}$, this additional complication is unnecessary.
\end{remark}\smallbreak

\section{Simulations}
\lab{sec6} To validate the theoretical results we consider a voltage source converter with rated power of $1000$ MVA, interfaced to a \mbox{$400$ kV} transmission grid operating at the nominal frequency of $f=50$ Hz, with nominal SCR of $3$ and an X/R ratio of $\rho=5$. Accordingly, the nominal parameters of the TE are given as follows:
$$
R=10.68\;\Omega,\;L=169.77\;\textrm{mH},\;E=326.60\;\textrm{kV}.
$$
We suppose that the transmission grid under consideration is characterized by a relatively high inertia and therefore we can safely assume that the frequency $\omega$ remains constant over the time-scale of interest. As for the converter side, we consider the scenario where the converter is operating at its rated power and that a standard synchronous reference frame PLL is deployed to quickly recover the actual value of the frequency. The proportional and integral gains for the PLL are set respectively to $\kappa_P=2\cdot 10^2$ and $\kappa_I=5\cdot 10^3$. The system is developed in the Matlab/Simulink simulation environment, release R2021b.\\
\begin{figure*}
\centering
\includegraphics[width=0.48\textwidth]{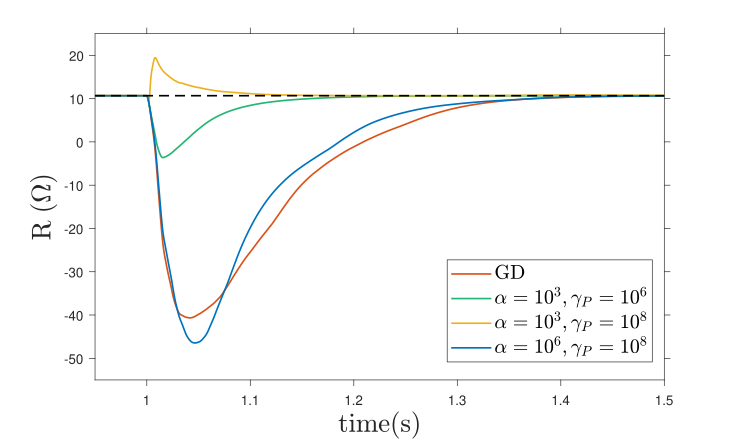}
\includegraphics[width=0.48\textwidth]{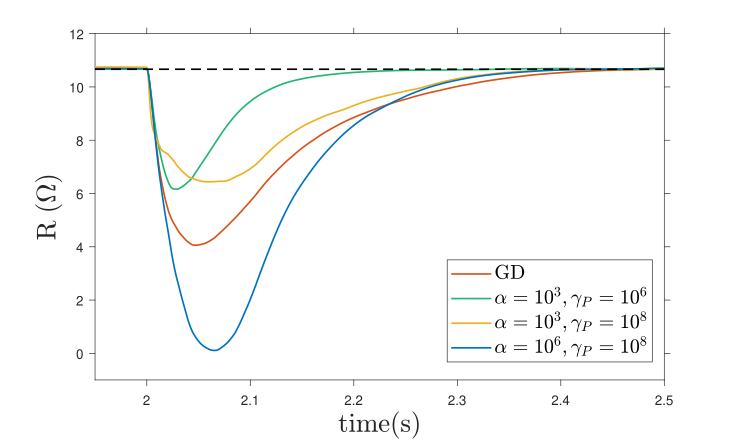}\\
\includegraphics[width=0.48\textwidth]{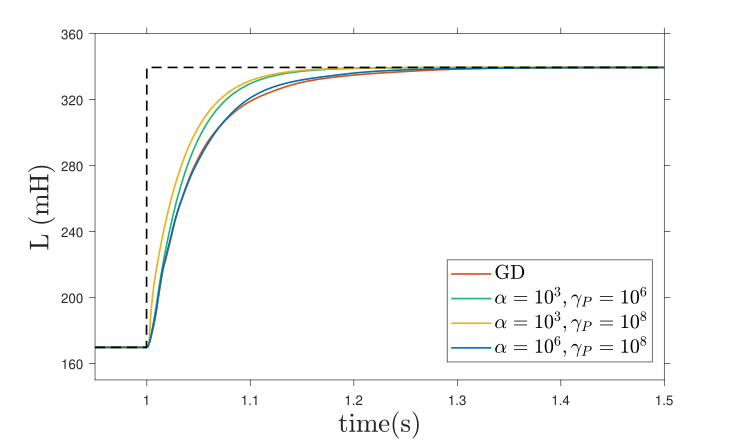}
\includegraphics[width=0.48\textwidth]{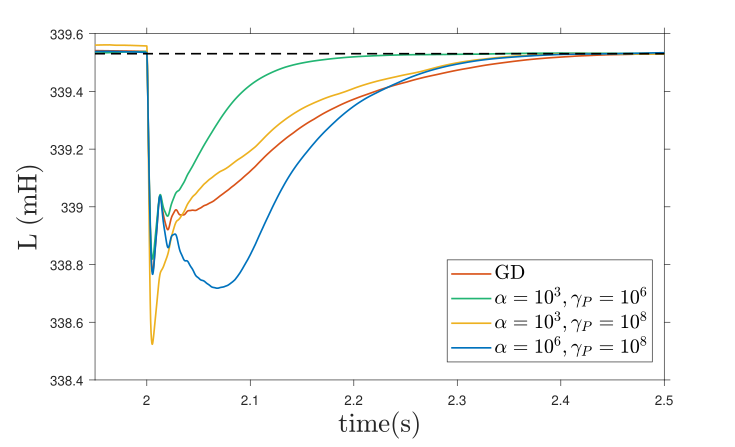}\\
\includegraphics[width=0.48\textwidth]{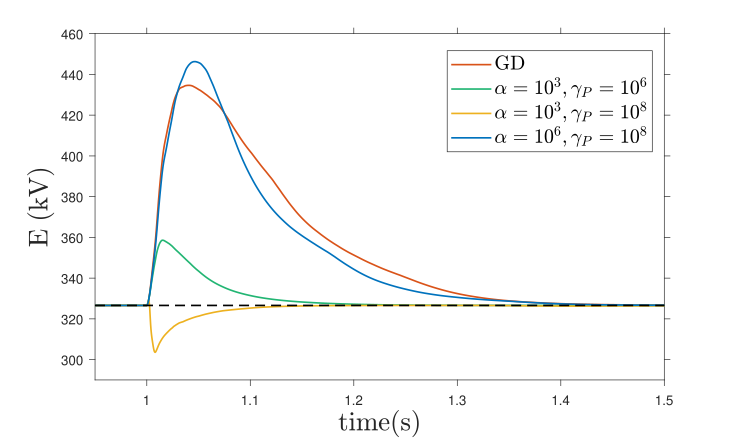}
\includegraphics[width=0.48\textwidth]{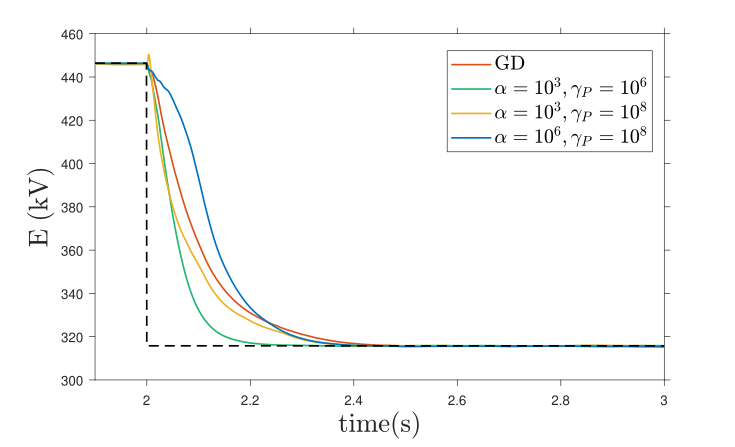}
     \caption{Estimates of the grid resistance, inductance and voltage amplitude (solid lines) obtained via the observer-based composite identifier \eqref{ident}, as compared to the corresponding, actual values (dotted lines), and following: a) a drop of the SCR from $3$ to $1.5$ at $T=1$ s (left); b) a drop of $10\%$ of the TE voltage source at $2T$ s (right).}  \label{fig:sims}
\end{figure*}
We evaluate the performance of the identifier following two types of perturbations occurring respectively at time $T=1$ s and $2T$ s. The first type of perturbation consists in a simultaneous change of the resistance and inductance of the TE, following a drop in the SCR from $3$ to $1.5$. The second type of perturbation takes the form of a drop of $10\%$ of the TE voltage source. Extensive simulations are realized for these two relevant scenarios, with the gains of the identifier \eqref{ident} set to fixed $\lambda=10^3$, $\gamma_I=10^8$ and variable $\alpha\in[0\;10^6]$, $\gamma_P\in[0\; 10^8]$. The obtained results for both scenarios and three illustrative gain pairs $(\alpha,\gamma_P)$ are reported in Fig. \ref{fig:sims}, where it is shown that for all such gains the identifier allows to correctly recover the actual values of the corresponding parameters. Unsurprisingly, we further observe that acceptable performances strongly depend on an appropriate selection of the gains of the identifier. A particular case is given whenever we pick  \mbox{$\alpha=\gamma_P=0$}, that is, the identifier coincides with a gradient descent (GD) algorithm, see also Remark~\ref{rem2}. We observe indeed that with this design, both a change of the impedance and of the voltage of the TE generate large peaks in the estimates provided by the identifier, resulting in convergence times superior to $350$ ms---performances that can be neatly improved with a different tuning of the gains. We thus conclude that the additional degrees of freedom provided by the observer-based composite identifier can be exploited to improve performances under all type of perturbations. \\
We complete this section briefly illustrating in Fig.~\ref{fig:GD} the performances of a gradient descent algorithm applied to the reduced LRE \eqref{lre-ab}, with tuning gain $\gamma=10^8$. In addition to the previously considered scenario, we evaluate the responses of the algorithm in case that a change of the X/R ratio, initially set to $\rho(0)=5$, may occur concurrently with the change of the SCR at time $T=1$ s, that is $\rho(t)=\{3,5,7\}$, for $t\geq T$. It is shown that, as long as the X/R ratio remains unchanged, the estimates fastly and asymptotically converge to their actual values---a fact that stems from Fig. \ref{fig:lambda} where we plot 
$$
\lambda_{\min}\Big\{\int_{0}^{t}\bfpsi_{f,\tt{ab}}(s)\bfpsi^\top_{f,\tt{ab}}(s)ds\Big\},
$$
which grows to infinity. As is well-known \cite[Proposition 4]{BARORT} the latter is a necessary condition for global convergence of the gradient estimator.  Nevertheless, whenever the ratio $\rho$ changes at $T=1\;s $, {\bf Assumption A5} is violated and large estimation errors are observed, suggesting that caution must be taken in employing such solution.

\begin{figure}
\centering
\includegraphics[width=0.48\textwidth]{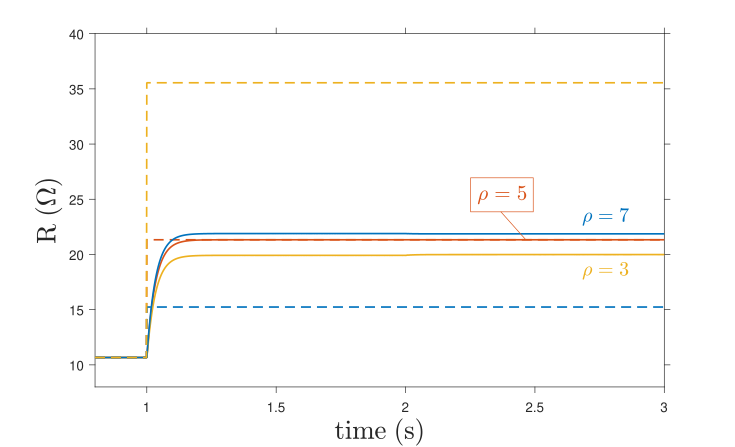}\\
\includegraphics[width=0.48\textwidth]{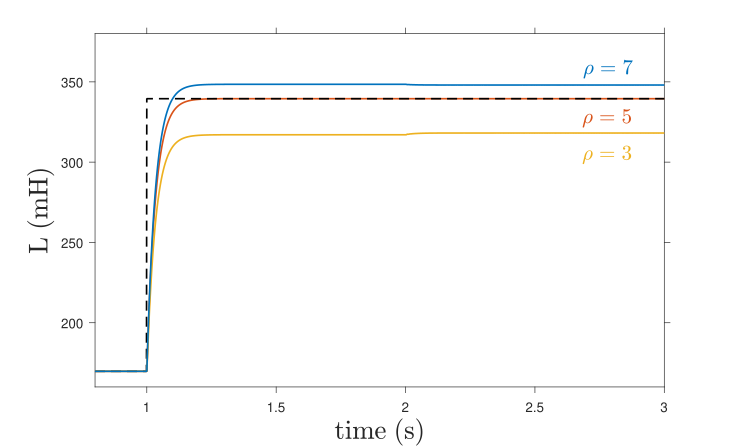}\\
\includegraphics[width=0.48\textwidth]{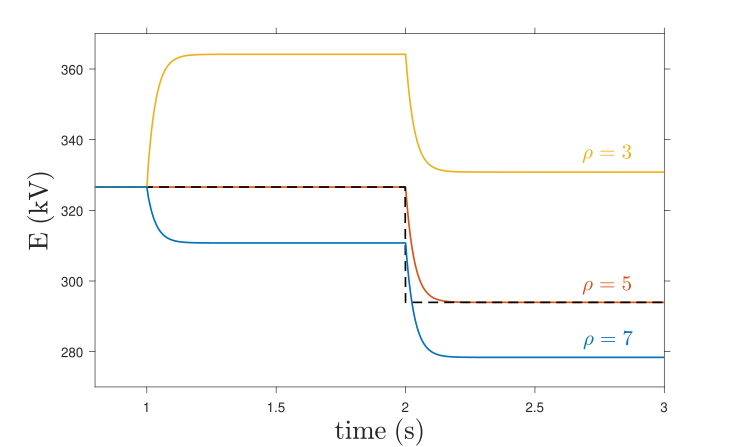}\\
     \caption{Estimates of the grid resistance, inductance and voltage amplitude (solid lines) obtained via a classical gradient descent algorithm applied to the LRE \eqref{lre-ab}, as compared to the corresponding, actual values (dotted lines), and following: a) a simultaneous change of the SCR from $3$ to $1.5$ and of the ratio $\rho$ from $5$ to $\{3,5,7\}$ at $T=1$ s; b) a drop of $10\%$ of the TE voltage source at $2T$ s.}  \label{fig:GD}
\end{figure}
\begin{figure}
    \centering
    \includegraphics[width=0.48\textwidth]{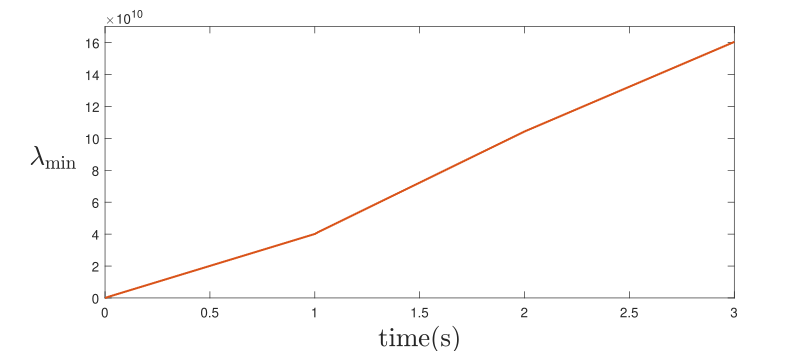}
    \caption{Minimum eigenvalue of the matrix $\int_{0}^{t}\bfpsi_{f,\tt{ab}}(s)\bfpsi^\top_{f,\tt{ab}}(s)ds$.}
    \label{fig:lambda}
\end{figure}

\section{Conclusions}\label{sec7}
In this paper we have addressed the problem of online identification of the parameters of a Th\'evenin equivalent grid, where this is adopted to describe a power system interfaced with a power electronic converter in closed-loop with an exponentially stabilizing controller. Based on this setting, we have derived a linear regression model for the system and next designed an observer-based composite identifier ensuring that the resulting estimates are bounded for all positive gains. An alternative, exponentially converging design is further proposed, assuming either a dominantly inductive grid or the \textit{a priori} knowledge of the X/R ratio---conditions that allow to generate a reduced linear regressor equation. The theoretical results are validated via simulations, which suggest that convergence and suitable performances can be enforced by an appropriate tuning of the gains.\\
Future works will explore solutions able to guarantee asymptotic convergence of the estimates to their actual values by lifting {\bf Assumption A5}. Validation of the obtained results on a detailed benchmark and related experiments are under progress and will be reported soon. 

\end{document}